\documentclass{article}

\usepackage[margin=1in]{geometry}
\usepackage{mathtools,amssymb,amsthm,bm}
\usepackage[psamsfonts]{eucal}
\usepackage{enumitem,setspace,microtype,threeparttable,float,afterpage,placeins,cancel,leftidx,breqn}
\usepackage[export]{adjustbox}
\usepackage[ruled,linesnumbered,lined]{algorithm2e}
\usepackage[dvipsnames]{xcolor}
\usepackage[most]{tcolorbox}
\usepackage{soul}
\usepackage{cite}
\usepackage{standalone}
\usepackage{fancyhdr}
\usepackage{graphicx}
\usepackage{subcaption}
\usepackage{fau}

\usepackage[hidelinks]{hyperref}
\usepackage[capitalize,noabbrev]{cleveref}

\newcounter{appsection}
\renewcommand{\theappsection}{\Alph{appsection}}

\crefname{appsection}{Appendix}{Appendices}
\Crefname{appsection}{Appendix}{Appendices}

\newcommand{\appsection}[2][]{%
    \refstepcounter{appsection}%
    \section*{\texorpdfstring{APPENDIX \theappsection \\ #2}{Appendix #2}}%
    \addcontentsline{toc}{section}{Appendix \theappsection: #2}%
    \ifx&#1&\else\label{#1}\fi
}

\let\originalleft\left
\let\originalright\right
\renewcommand{\left}{\mathopen{}\mathclose\bgroup\originalleft}
\renewcommand{\right}{\aftergroup\egroup\originalright}

\newtheoremstyle{indentedplain}
  {3pt}      
  {3pt}      
  {\rmfamily}
  {1em}      
  {\bfseries}
  {.}        
  {0.5em}    
  {}         

\newtheoremstyle{indentedremark}%
  {3pt}{3pt}{}
  {1em}{\bfseries}{.}{0.5em}{}

\theoremstyle{indentedplain}
\newtheorem{theorem}{Theorem}

\newtheorem{lemma}{Lemma}

\newtheorem{definition}{Definition}
\newtheorem{fact}{Fact}

\theoremstyle{indentedremark}

\allowdisplaybreaks

\newcounter{enumA}
\newcounter{enumC}
\newcounter{enumO}

\newlist{enumA}{enumerate}{1}
\setlist[enumA,1]{%
  label=(A\arabic*),
  leftmargin=1cm,
  before=\setcounter{enumAi}{\value{enumA}},
  after=\setcounter{enumA}{\value{enumAi}}
}

\newlist{enumC}{enumerate}{1}
\setlist[enumC,1]{%
  label=(C\arabic*),
  leftmargin=1cm,
  before=\setcounter{enumCi}{\value{enumC}},
  after=\setcounter{enumC}{\value{enumCi}}
}

\newlist{enumO}{enumerate}{1}
\setlist[enumO,1]{%
  label=(O\arabic*),
  leftmargin=1cm,
  before=\setcounter{enumOi}{\value{enumO}},
  after=\setcounter{enumO}{\value{enumOi}}
}

\crefname{equation}{}{}

\DeclareMathAlphabet{\mathcal}{OMS}{cmsy}{m}{n} 
    
\title{Higher-Order Gravitational Models:\\
A Tutorial on Spherical Harmonics and the Newtonian Model}

\author{Felipe Arenas-Uribe%
\thanks{University of Kentucky, \texttt{felipearur@uky.edu}}}

\date{January 2026}

\begin{document}

\maketitle

Gravitational interactions play a central role in the modeling, analysis, and control of space systems. While the Newtonian point-mass approximation provides useful insight and is sufficient for many preliminary analyses, real celestial bodies deviate significantly from spherical symmetry. Oblateness, localized mass concentrations, and higher-order shape irregularities introduce perturbations that can substantially influence spacecraft trajectories, particularly for low-altitude or long-duration missions. In such regimes, inaccuracies in the gravitational model may lead to cumulative orbit prediction errors, degraded navigation performance, and increased control effort. Higher-order gravitational field models are therefore essential for achieving the levels of accuracy required in modern mission design and operation.

The aim of this article is to serve as a tutorial introduction to spherical harmonic gravity models, clarifying their theoretical origins and the assumptions underlying their use. We begin by introducing the concept of a gravitational field and the conditions under which it may be represented by a scalar potential. Building on this foundation, higher-order gravitational field models are derived as solutions to the Laplace equation, leading to the widely used spherical harmonic formulation. This approach provides a systematic representation of gravitational effects arising from non-uniform mass distributions. The influence of higher-order terms on orbital dynamics is then illustrated through examples of perturbed systems, including satellites in Low Earth Orbit and spacecraft operating near irregularly shaped asteroids, demonstrating how deviations from the point-mass assumption arise in realistic mission scenarios.

\section{Equation of Motion and the Gravitational Field}

    We begin by making the modeling assumption that we live in a three-dimensional Euclidean space. While General Relativity shows this is not strictly true, it provides a convenient mathematical framework to work with. Let $\mathcal{V} \subset \BBR^3$ be a compact, simply connected set representing the points that compose the main body, with a mass density distribution $\rho: \mathcal{V} \to \BBR^+$. Let $o_{I} \in \mathcal{V}$ denote the center of mass of the main body, and let the total mass $M$ be
    \begin{equation}\label{eq:mass_integral}
        M \triangleq \int_{\mathcal{V}} \rho(p) \, dV.
    \end{equation}
    
    Let $\mathcal{S} \triangleq \BBR^3 \setminus \mathcal{V}$ denote the space outside the main body, where we wish to model gravitational effects. Consider a particle $o_{B} \in \mathcal{S}$ with mass $m > 0$, which we will treat as the test body. That is, we study the gravitational force exerted by the main body on the test body under the assumption that $M \gg m$.
    
    Next, we define a reference frame in which to study the motion of the test body. Let $\SF_{\rm I}$ be an inertial frame with origin at $o_{\rm I}$ and orthonormal basis vectors $\hat{I}$, $\hat{J}$, and $\hat{K}$. The position of $o_{\rm B}$ relative to $o_{\rm I}$ is $\mathbf{p} \triangleq x \hat{I} + y \hat{J} + z \hat{K}$, and the acceleration of $o_{\rmB}$ relative to $o_{\rmI}$ with respect to $\SF_{\rmI}$ is $\mathbf{a} \triangleq ^{\SF_{\rmI}}\ddot{(\mathbf{r})}$.
    
    Let $\mathbf{F}_{\rm g} = F_x \hat{I} + F_y \hat{J} + F_z \hat{K}$ denote the gravitational force of the main body on the test body. This force depends on position and scales linearly with the test-body mass. We also resolve the position, acceleration, and force vectors in the inertial frame as $p \triangleq [\mathbf{p}]_{\SF_{\rm I}}$, $a \triangleq [\mathbf{a}]_{\SF_{\rm I}}$ and $F_{\rm g} \triangleq [\mathbf{F}_{\rm g}]_{\SF_{\rm I}}$ respectively.
    
    Then, Newton's second law yields
    \begin{equation*}
        m \mathbf{a} = \mathbf{F}_{\rm g}(\mathbf{p}).
    \end{equation*}
    And resolving in the inertial frame yields
    \begin{equation*}
        \ddot{p}(t) = \frac{1}{m} F_{\rm g}\bigl(p(t)\bigr),
    \end{equation*}
    which is a second-order ordinary differential equation (ODE) governing the translational motion of the test body. Since the force $F_{\rm g}$ is linearly proportional to the mass of the test body, the resulting acceleration is independent of $m$. This motivates the introduction of the gravitational acceleration field, defined by $F(r) \triangleq \frac{1}{m} F_{\rm g}(r)$. As a result, the equation of motion may be written as
    \begin{equation}\label{eq:eq_mot}
        \ddot{p}(t) = F\bigl(p(t)\bigr),
    \end{equation}
    which emphasizes that gravitational motion is governed by a vector field defined over space. This formulation is particularly convenient for analytical gravity models, as it allows the gravitational environment to be characterized independently of the test body. In the next section we will explore how we can model this vector field.
    
    \begin{figure}[ht]
        \centering
        \includegraphics[width=0.3\linewidth]{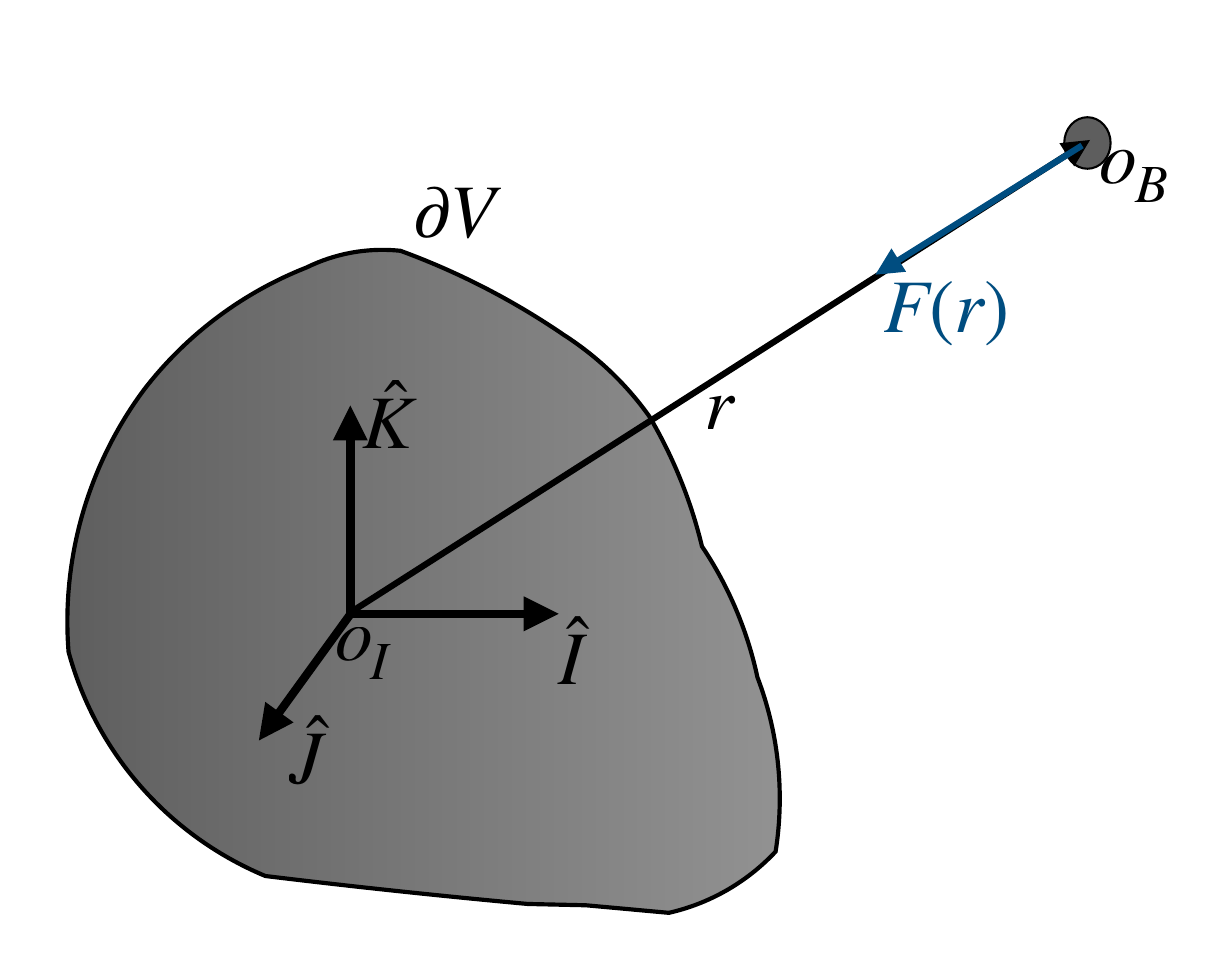}
        \caption{Inertial frame, main body and test body under consideration.}
        \label{fig:frame}
    \end{figure}

\section{Gravitational Field as a Conservative Field}

    A common approach in mechanics to study the motion of a particle in a vector field is to analyze the evolution of its energy and the interaction between external forces and that energy. Since $F$ depends only on position, the energy associated with the vector field in \cref{eq:eq_mot} is purely potential. Hence, we seek to develop a model of the gravitational force in terms of a potential energy function. The following results are motivated by \cite[Section~1.13]{arfken_mathematical_2013}. We begin by introducing several definitions from potential theory, together with their physical interpretations.
    
    \begin{definition}
        A vector field $F:\mathcal{S} \to \BBR^3$ is \emph{conservative} if there exists a scalar potential function $U:\mathcal{S} \to \BBR$ such that
        \begin{equation*}
            F(p) = - \nabla U(p).
        \end{equation*}
    \end{definition}
    
    Physically, a conservative force field is one whose action can be fully described by a scalar potential function. The force at any point is directed toward decreasing potential, and all work performed by the force corresponds to changes in stored energy rather than dissipation.
    
    \begin{definition}\label{def:irrotational}
        A vector field $F:\mathcal{S} \to \BBR^3$ is \emph{irrotational} if
        \begin{equation*}
            \nabla \times F(p) = 0.
        \end{equation*}
    \end{definition}
    
    Physically, irrotationality means that the field exhibits no local circulation. A particle moving under such a force experiences no tendency to rotate about any point. This is a local condition, describing the behavior of the field in an arbitrarily small neighborhood. To connect this local property to global energy considerations, we introduce the following definition.
    
    \begin{definition}\label{def:conservation_energy}
        A vector field $F:\mathcal{S} \to \BBR^3$ is said to have \emph{vanishing circulation} if
        \begin{equation*}
            \oint F(p) \cdot dp = 0
        \end{equation*}
        for every closed trajectory in $\mathcal{S}$.
    \end{definition}
    
    This definition formalizes the physical principle that the work performed between two points depends only on the endpoints and not on the path taken, which is precisely the hallmark of energy conservation. The following fact establishes a fundamental link between local field structure and global energetic behavior. The proof is in \cite[Section 1.13]{arfken_mathematical_2013}.
    
    \begin{fact} \label{fact:irrotational_conservation}
        Consider a vector field $F:\mathcal{S} \to \BBR^3$. Then $F$ is irrotational if and only if it has vanishing circulation.
    \end{fact}
    
    This equivalence highlights that the absence of local rotational effects guarantees that no energy can be gained or lost by traversing a closed loop. We can now state an existence result for conservative force fields. The proof is in \cite[Section 1.13]{arfken_mathematical_2013}.
    
    \begin{fact} \label{fact:potential_existence}
        Consider the vector field $F:\mathcal{S} \to \BBR^3$. The following statements hold:
        \begin{enumerate}
            \item $F$ is conservative if and only if $F$ is irrotational.
            \item $F$ is conservative if and only if $F$ conserves energy.
        \end{enumerate}
    \end{fact}
    
    This result is fundamental for modeling gravitational fields. It asserts that the existence of a gravitational potential is not an additional assumption, but a direct consequence of the field being irrotational and having vanishing circulation. As a result, gravitational forces can be equivalently described through a scalar potential, greatly simplifying analysis. This motivates our first modeling assumption on the gravitational field. 
    
    \begin{enumA}
        \item The gravitational vector field is irrotational or has vanishing circulation. \label{assum:potential_existence}
    \end{enumA} 

    Then, we leverage on this assumption to formalize the next result for the existence of a gravitational potential function. The proof follows directly from \cref{fact:irrotational_conservation} and \cref{fact:potential_existence}.
    
    \begin{lemma}\label{lemma:grav_potential_existence}
        Consider the vector field $F:\mathcal{S} \to \BBR^3$ and assume that \ref{assum:potential_existence} is satisfied. Then, there exists a potential function $U: \mathcal{S} \to \BBR$ such that
        \begin{equation}\label{eq:conservative_potential}
            F(p) = - \nabla U(p).
        \end{equation}
    \end{lemma}

\section{Modeling the Gravitational Potential}

    Now that we have shown that there exists a gravitational potential function, we focus on modeling this function. We begin by establishing a relation between the gravitational force and the mass distribution of the main body. To this end, we invoke Gauss's law for gravity, which states that the flux of the gravitational field through a closed surface is proportional to the mass enclosed.
    
    \begin{enumA}
        \item The gravitational field $F$ satisfies \label{ass:Gauss_law}
        \begin{equation}\label{eq:Gauss_law_int}
            \oint_{\partial \mathcal{V}} F \cdot \hat{n} dS = -4 \pi G M.
        \end{equation}
    \end{enumA}
    
    This is where the gravitational constant $G = 6.6743 × 10-11 \frac{\rm m^3}{\text{kg } \mathrm{s}^2}$ shows up. 
    
    \begin{lemma}\label{lemma:Gauss}
        Assume that \ref{ass:Gauss_law} is satisfied. Then,
        \begin{equation}\label{eq:Gauss_law_diff}
            \nabla \cdot F(p) = -4 \pi G \rho(p).
        \end{equation}
    \end{lemma}
    
    \begin{proof}\noindent
        Applying the Divergence Theorem on the left hand side of \ref{eq:Gauss_law_diff} yields
        \begin{equation}\label{eq:divergence_thm_on_grav}
            \oint_{\partial \mathcal{V}} F \cdot \hat{n} dS = \int_{\mathcal{V}}\nabla \cdot F d \mathcal{V}.
        \end{equation}
    
        Substituting \cref{eq:mass_integral,eq:divergence_thm_on_grav} in \cref{eq:Gauss_law_int} yields
        \begin{equation*}
            \int_{\mathcal{V}}\nabla \cdot F d \mathcal{V} = -4 \pi G \int_{\mathcal{V}} \rho dV,
        \end{equation*}
        which is equivalent to
        \begin{equation*}
            \int_{\mathcal{V}}(\nabla \cdot F) d \mathcal{V} = \int_{\mathcal{V}} (-4 \pi G  \rho) dV.
        \end{equation*}
        For this equality to hold, the integrands must be equal. Thus,
        \begin{equation*}
            \nabla \cdot F = -4 \pi G  \rho.
        \end{equation*}
    \end{proof}

    \begin{theorem}\label{thm:Grav_Field_Laplace_Eq}
        Consider the gravitational field $F: \mathcal{S} \to \BBR^3$. Assume that \ref{assum:potential_existence} and \ref{ass:Gauss_law} are satisfied. Then, there exists a gravitational potential $U: \mathcal{S} \to \BBR$ such that for all $p \in \mathcal{S}$,
        \begin{equation}\label{eq:Laplace_equation}
            \nabla^2U(p) = 0.
        \end{equation}
    \end{theorem}

    \begin{proof}
        It follows from \cref{lemma:grav_potential_existence} that there exists a scalar potential $U: \mathcal{S} \to \BBR$ such that
        \begin{equation*}
            F(p) = - \nabla U(p).
        \end{equation*}
        Next, it follows from \cref{lemma:Gauss} that the gravitational field satisfies
        \begin{equation*}
            \nabla \cdot F(p) = -4 \pi G \rho(p).
        \end{equation*}
        Note that for all $p \in \mathcal{S}$, we have $\rho(p) = 0$. Thus, $\nabla \cdot F(p) = 0$. And substituting $F = - \nabla U$ yields
        \begin{equation*}
            \nabla \cdot (-\nabla U(p)) = -\nabla^2 U(p) = 0.
        \end{equation*}
        Hence, the gravitational potential $U$ satisfies
        \begin{equation*}
            \nabla^2 U(p) = 0.
        \end{equation*}
    \end{proof}

    This result shows that the gravitational potential is governed by \textit{Laplace’s equation} \cref{eq:Laplace_equation}, which is a partial differential equation (PDE). As a consequence, modeling the gravitational field reduces to a well-posed boundary-value problem of Laplace’s equation. Functions satisfying this equation are known as \textit{harmonic functions}, and their mathematical structure will form the basis of the gravitational models developed in the remainder of this work.

\section{Exterior solution of Laplace's equation}

    To solve for the gravitational potential outside the main body, we formulate Laplace's equation \cref{eq:Laplace_equation} in spherical coordinates. Consider the spherical coordinates: radius $r \in [0,\infty)$, inclination $\theta \in [0,\pi]$, and azimuth $\varphi \in [0,2\pi)$, defined as $r \triangleq \|p\|_2$, $\theta \triangleq \arccos\!\left(\frac{z}{\|p\|_2}\right)$ and $\varphi \triangleq \operatorname{atan2}(y,x)$. In these coordinates, Laplace's equation takes the form
    \begin{equation}\label{eq:spherical_Laplace}
        \frac{1}{r^2}\frac{\partial}{\partial r}
        \left( r^2 \frac{\partial U}{\partial r} \right)
        + \frac{1}{r^2 \sin\theta}\frac{\partial}{\partial \theta}
        \left( \sin\theta \frac{\partial U}{\partial \theta} \right)
        + \frac{1}{r^2 \sin^2\theta}
        \frac{\partial^2 U}{\partial \varphi^2}
        = 0.
    \end{equation}

    The following result characterizes the solutions to Laplace's equation. The proof is based on \cite[Section 9.3]{arfken_mathematical_2013}.

    \begin{theorem}\label{thm:Laplace_solution}
        Consider Laplace's equation in spherical coordinates \cref{eq:spherical_Laplace}. Assume that solutions admit separation of variables for each variable. Then, the general solution of \cref{eq:spherical_Laplace} is
        \begin{equation}\label{eq:spherical_Laplace_gen_solution}
            U(r,\theta,\varphi)
            =
            \sum_{\ell=0}^{\infty}
            \sum_{m=-\ell}^{\ell}
            R_\ell(r)\,\Theta_{\ell,m}(\theta)\,\Phi_m(\varphi),
        \end{equation}
        where $\ell \in \mathbb{Z}_{\ge 0}$ and $m \in \mathbb{Z}$ such that $|m| \le \ell$. And $R_\ell(r)$, $\Theta_{\ell,m}(\theta)$ and $\Phi_m(\varphi)$ satisfy the ordinary differential equations
        \begin{align}
            &\frac{1}{r^2}
            \frac{d}{dr}\!\left( r^2 \frac{dR_\ell(r)}{dr} \right)
            - \frac{\ell(\ell+1)}{r^2} R_\ell(r)
            = 0,
            \label{eq:radial_ode}\\
            &\frac{1}{\sin\theta}
            \frac{d}{d\theta}\!\left( \sin\theta \frac{d\Theta_{\ell m}(\theta)}{d\theta} \right)
            - \frac{m^2}{\sin^2\theta}\Theta_{\ell m}(\theta)
            + \ell(\ell+1)\Theta_{\ell m}(\theta)
            = 0,
            \label{eq:polar_ode}\\
            &\frac{d^2\Phi(\varphi)}{d\varphi^2} + m^2 \Phi(\varphi) = 0.\label{eq:azimuth_ode}
        \end{align}
    \end{theorem}

    \begin{proof}
        We seek separable solutions of the form
        \begin{equation}\label{eq:general_solution}
            U(r,\theta,\varphi) = R(r)\,\Theta(\theta)\,\Phi(\varphi),
        \end{equation}
        where $R(r)$ captures the radial dependence and $\Theta(\theta)\Phi(\varphi)$ captures the angular dependence. Substituting \cref{eq:general_solution} into \cref{eq:spherical_Laplace} and dividing by $R(r)\Theta(\theta)\Phi(\varphi)$ yields
        \begin{equation}\label{eq:linear_spherical_Laplace}
            \frac{1}{r^2 R(r)}
            \frac{d}{dr}\!\left( r^2 \frac{dR(r)}{dr} \right)
            + \frac{1}{r^2 \Theta(\theta)\sin\theta}
            \frac{d}{d\theta}\!\left( \sin\theta \frac{d\Theta(\theta)}{d\theta} \right)
            + \frac{1}{r^2 \sin^2\theta\,\Phi(\varphi)}
            \frac{d^2\Phi(\varphi)}{d\varphi^2}
            = 0.
        \end{equation}
        
        Note that all derivatives are now ordinary derivatives. Multiplying \cref{eq:linear_spherical_Laplace} by $r^2 \sin^2\theta$ gives
        \begin{equation}\label{eq:first_sep}
            \frac{1}{\Phi(\varphi)}
            \frac{d^2\Phi(\varphi)}{d\varphi^2}
            =
            -\sin^2\theta
            \frac{1}{R(r)}
            \frac{d}{dr}\!\left( r^2 \frac{dR(r)}{dr} \right)
            - \frac{\sin\theta}{\Theta(\theta)}
            \frac{d}{d\theta}\!\left( \sin\theta \frac{d\Theta(\theta)}{d\theta} \right).
        \end{equation}
        
        Since the left-hand side depends only on $\varphi$ while the right-hand side depends only on $r$ and $\theta$, both sides must be equal to a constant. Since $\varphi$ is an azimuthal angle, the solution $\Phi$ must be single-valued and $2\pi$-periodic, i.e., $\Phi(\varphi + 2\pi) = \Phi(\varphi)$. Thus, define the separation constant as $\lambda_1 \triangleq -m^2$, where $m \in \mathbb{Z}$. The resulting separated equations are
        \begin{align}
            &\frac{d^2\Phi(\varphi)}{d\varphi^2} + m^2 \Phi(\varphi) = 0,\\
            &\frac{1}{R(r)}\frac{d}{dr}\!\left( r^2 \frac{dR(r)}{dr} \right)
            =
            -\frac{1}{\Theta(\theta)\sin\theta}
            \frac{d}{d\theta}\!\left( \sin\theta \frac{d\Theta(\theta)}{d\theta} \right)
            + \frac{m^2}{\sin^2\theta}.
            \label{eq:second_sep}
        \end{align}
        
        The remaining equation separates $r$ and $\theta$ by equating each side of \cref{eq:second_sep} to a constant $\lambda_2$. To ensure regularity of the polar solution $\Theta(\theta)$ at $\theta = 0$ and $\theta = \pi$, the separation constant is chosen as $\lambda_2 \triangleq \ell(\ell+1)$, where $\ell \in \mathbb{Z}_{\ge 0}$ and $|m| \le \ell$. The resulting ordinary differential equations are
        \begin{align}
            &\frac{1}{\sin\theta}
            \frac{d}{d\theta}\!\left( \sin\theta \frac{d\Theta_{\ell m}(\theta)}{d\theta} \right)
            - \frac{m^2}{\sin^2\theta}\Theta_{\ell m}(\theta)
            + \ell(\ell+1)\Theta_{\ell m}(\theta)
            = 0,\\
            &\frac{1}{r^2}
            \frac{d}{dr}\!\left( r^2 \frac{dR_\ell(r)}{dr} \right)
            - \frac{\ell(\ell+1)}{r^2} R_\ell(r)
            = 0.
        \end{align}
        
        Consequently, the solution \cref{eq:general_solution} to the partial differential equation \cref{eq:linear_spherical_Laplace} is now given as the solution to the system of ordinary differential equations \cref{eq:azimuth_ode,eq:polar_ode,eq:radial_ode}, which takes the form
        \begin{equation}
            U(r,\theta,\varphi)
            =
            \sum_{\ell=0}^{\infty}
            \sum_{m=-\ell}^{\ell}
            R_\ell(r)\,\Theta_{\ell,m}(\theta)\,\Phi_m(\varphi),
        \end{equation}
    
    \end{proof}
    
    \Cref{thm:Laplace_solution} suggests that solving Laplace's equation reduces to solving a sequence of three ordinary differential equations. This formulation provides a systematic framework for constructing gravitational potentials. We will proceed to find solutions for the radial and angular components in the following subsections.

\subsection{The Radial Solution}

    Equation \eqref{eq:radial_ode} is an Euler--Cauchy equation, which admits power--law solutions of the form $R_\ell(r)=r^k$. Substitution into \eqref{eq:radial_ode} yields the characteristic equation
    \begin{equation*}
        k(k+1) - \ell(\ell+1) = 0,
    \end{equation*}
    whose roots are $k=\ell$ and $k=-(\ell+1)$. Since $R_\ell(r) \to 0$ as $r \to \infty$, then the general radial solution is
    \begin{equation}\label{eq:radial_exterior}
        R_\ell(r) = \frac{B_\ell}{r^{\ell+1}}.
    \end{equation}
    where $B_\ell \in \mathbb{R}$ is an integration constant, directly related to the mass moments of the body.
    
\subsection{The Angular Solution}

The angular dependence of the gravitational potential \cref{eq:spherical_Laplace_gen_solution} can be separated into an azimuthal component $\Phi_m(\varphi)$ and a polar component $\Theta_{\ell m}(\theta)$, where $\ell \in \mathbb{Z}_{\ge 0}$ and $m \in \mathbb{Z}$ with $|m| \le \ell$. 

The solution to the azimuthal equation \cref{eq:azimuth_ode} is
\begin{equation}\label{eq:azimuth_solution}
    \Phi_m(\varphi) =
    \begin{cases}
        \cos(m\varphi), & m = 0,1,2,\dots, \\[0.2em]
        \sin(m\varphi), & m = 1,2,\dots,
    \end{cases}
\end{equation}
which forms an orthonormal set over $\varphi \in [0,2\pi)$.  

The polar equation \cref{eq:polar_ode}, known as the associated Legendre equation \cite[Chapter 12.5]{arfken_mathematical_2013}, has solutions given by the associated Legendre functions $P_\ell^m(\cos\theta)$. These can be expressed in terms of the standard Legendre polynomials $P_\ell(x)$ via Rodrigues' formula:
\begin{equation}\label{eq:polar_solution}
    \Theta_{\ell m}(\theta) = P_\ell^m(\cos\theta) \;\triangleq\; 
    (-1)^m (1-\cos^2\theta)^{m/2} 
    \frac{d^m}{d(\cos\theta)^m} 
    \left[ \frac{1}{2^\ell \ell!} \frac{d^\ell}{d(\cos\theta)^\ell} \left( (\cos\theta)^2 - 1 \right)^\ell \right],
\end{equation}
which are orthogonal over $\theta \in [0,\pi]$ with respect to the measure $\sin\theta\,d\theta$.

Combining these two components, we define $Y_\ell^m(\theta, \varphi) \triangleq \Theta_{\ell m}(\theta)\Phi_m(\varphi)$, where
\begin{equation}\label{eq:real_spherical_harmonics}
    Y_\ell^m(\theta, \varphi) =
    \begin{cases}
        P_\ell(\cos\theta), & m = 0, \\[0.4em]
        \sqrt{2}\,P_\ell^m(\cos\theta)\,\cos(m\varphi),
        & m = 1,2,\dots,\ell, \\[0.4em]
        \sqrt{2}\,P_\ell^{|m|}(\cos\theta)\,\sin(|m|\varphi),
        & m = -1,-2,\dots,-\ell.
    \end{cases}
\end{equation}

The functions $Y_\ell^m(\theta, \varphi)$ are the real solution to the angular component of Laplace’s equation in spherical coordinates \cref{eq:spherical_Laplace}. This family of functions are orthogonal over the surface of the sphere, which motivates the name \textbf{spherical harmonics}, as it reflects their mathematical and physical origin.

\subsection{Exterior General Solution}

The results of the previous subsections show that any solution of Laplace’s
equation in the exterior domain can be written as a superposition of separable
modes consisting of a radial power law and an angular spherical harmonic.
Specifically, combining the admissible exterior radial solution
\eqref{eq:radial_exterior} with the real spherical harmonics
\eqref{eq:real_spherical_harmonics}, the general exterior harmonic potential
takes the form
\begin{equation}\label{eq:general_ext_solution}
    U(r,\theta,\varphi)
    =
    \sum_{\ell=0}^{\infty}
    \sum_{m=-\ell}^{\ell}
    \frac{B_{\ell m}}{r^{\ell+1}}
    Y_\ell^m(\theta,\varphi).
\end{equation}

We may now derive the simplest gravitational potential and the basis for higher-order gravitational potential models, the Newtonian potential.

    \begin{theorem}\label{thm:Newtonian_potential}
        Assume that the main body is spherically symmetric with a uniform mass
        distribution. Then, for all $r \in \mathcal{S}$ the gravitational potential $U:\mathcal{S} \to \BBR$ is
        \begin{equation}
            U_{\rm N}(r) = \frac{GM}{r}.
        \end{equation}
    \end{theorem}
    
    \begin{proof}
        It follows from \cref{thm:Grav_Field_Laplace_Eq} that for all
        $r \in \mathcal{S} \setminus \mathcal{V}$ the gravitational potential
        satisfies Laplace’s equation,
        \begin{equation*}
            \nabla^2 U(r,\theta,\varphi) = 0 .
        \end{equation*}
        with the exterior general solution of Laplace’s equation in spherical coordinates \cref{eq:general_ext_solution}.

        Next, it follows from \cref{thm:Laplace_solution} that the general solution of Laplace's equation takes the form
        \begin{equation*}
            U(r,\theta,\varphi)
            =
            \sum_{\ell=0}^{\infty}
            \sum_{m=-\ell}^{\ell}
            \frac{B_{\ell m}}{r^{\ell+1}}
            Y_\ell^m(\theta,\varphi).
        \end{equation*}
        
        Since the mass distribution of the main body is spherically symmetric, the
        gravitational potential is invariant under arbitrary rotations. Consequently, the potential depends only on the radial coordinate $r$. Note that Spherical harmonics with $\ell \ge 1$ are not invariant under rotations, whereas the $\ell = 0$ harmonic,
        \begin{equation*}
            Y_0^0(\theta,\varphi) = 1,
        \end{equation*}
        is constant on the sphere. Therefore, rotational invariance implies that
        the solution reduces to the $\ell = 0$ term, which is reduced to the radial component. The radial component $R_0(r)$ thus satisfies
        \begin{equation}\label{eq:radial_ode_l0}
            \frac{1}{r^2}
            \frac{d}{dr}\!\left( r^2 \frac{dR_0(r)}{dr} \right) = 0.
        \end{equation}
        Integrating \cref{eq:radial_ode_l0} yields the general exterior radial solution
        \begin{equation*}
            R_0(r) = C_1 + \frac{C_2}{r}.
        \end{equation*}
        Let $C_1 = 0$ and $C_2 = GM$. Then, the Newtonian gravitational potential is
        \begin{equation*}
            U_{\rm N}(r) = \frac{GM}{r}.
        \end{equation*}
    \end{proof}

Note that the Newtonian potential is equivalently the monopole term $U_{0,0}(r) = U_{\rm N}(r)$. Substituting this result back into \eqref{eq:general_ext_solution}, we can isolate the monopole term from the higher-order contributions
\begin{equation}\label{eq:general_ext_solution_monopole}
    U(r,\theta,\varphi)
    =
    \frac{GM}{r}
    +
    \sum_{\ell=1}^{\infty}
    \sum_{m=-\ell}^{\ell}
    \frac{B_{\ell m}}{r^{\ell+1}}
    Y_\ell^m(\theta,\varphi),
\end{equation}
which clearly distinguishes the dominant point-mass potential from the multipole contributions that encode the body's deviations from spherical symmetry. Note that each term of degree $\ell$ decays as $r^{-(\ell+1)}$, so higher-degree harmonics have diminishing influence with distance. Which motivates the study of higher-order gravitational effects as pertubations to the Newtonian model as done in some of the classic textbooks \cite[Chapter 12.5]{curtis_orbital_2021}, \cite[Chapter 11]{schaub_analytical_2018} or \cite[Chapter 8.6]{vallado_fundamentals_2022}.

\section{Computational Methods for the Gravitational Potential}

Despite its expressiveness, the general solution~\cref{eq:general_ext_solution_monopole} is often challenging to compute for several reasons. First, the coefficients $B_{\ell m}$ span multiple orders of magnitude, even for low-degree and low-order harmonics, making their relative contributions difficult to assess. Second, although the harmonic coefficients are physically meaningful, they are difficult to determine experimentally with sufficient accuracy. Third, the computation of spherical harmonics relies on recursive relations involving derivatives of lower-order terms, which can lead to numerical instability and increased computational cost. In this section, we present a set of simplifications and computational strategies designed to address each of these challenges.

\subsection{Exterior Geometric Solution}

    For the higher-order multipole terms with $\ell \ge 1$, the coefficients 
    $B_{\ell m}$ carry physical dimensions that scale with $r^\ell$, 
    which complicates comparison across degrees and obscures their relative 
    contributions. To address this, we introduce a reference length scale $R>0$, 
    chosen such that the entire mass distribution is contained within a sphere of 
    radius $R$ centered at the origin. In planetary applications, $R$ is typically taken to be the equatorial radius of the body. Then, define the normalized coefficient $B_{\ell m}$ as
    \begin{equation}\label{eq:Blm_scaling}
        \bar{B}_{\ell m} \triangleq \frac{B_{\ell m}}{GM\,R^{\ell}},
    \end{equation}
    where $\bar B_{\ell m}\in\mathbb{R}$ are dimensionless constants. This formulation facilitates both physical interpretation and computational modeling of the exterior gravitational field. Substituting \eqref{eq:Blm_scaling} into \eqref{eq:general_ext_solution_monopole} yields
    \begin{equation}\label{eq:dimensionless_ext_solution}
        U(r,\theta,\varphi)
        =
        \frac{GM}{r} \left [ 1 +
        \sum_{\ell=1}^{\infty}
        \sum_{m=-\ell}^{\ell}
        \bar B_{\ell m}
        \left( \frac{R}{r} \right)^{\ell}
        Y_\ell^m(\theta,\varphi) \right ],
    \end{equation}
    
    Note that \cref{eq:dimensionless_ext_solution} is a geometric series with ratio $R/r$, which implies that the series converges if and only if $r>R$. This property defines the \emph{Brillouin sphere} condition (see \cref{fig:Brillouin}) and establishes the domain of validity for the exterior spherical harmonic expansion.

    \begin{figure}
        \centering
        \includegraphics[width=0.3\linewidth]{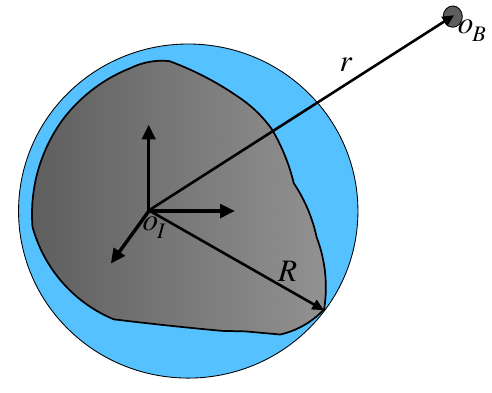}
        \caption{Brillouin sphere around the main body given by the normalizing parameter $R$. The series in \cref{eq:dimensionless_ext_solution} converges outside the blue sphere.}
        \label{fig:Brillouin}
    \end{figure}
    
    We express the dimensionless coefficients $\bar B_{\ell m}$ in terms of two independent real coefficients, $C_{\ell m}$ and $S_{\ell m}$, which weight the even (cosine) and odd (sine) azimuthal components, respectively
    \begin{equation*}
        \bar B_{\ell m}
        \;\longleftrightarrow\;
        \begin{cases}
            C_{\ell m}, & m = 0,1,2,\dots,\ell, \\[0.3em]
            S_{\ell |m|}, & m = -1,-2,\dots,-\ell.
        \end{cases}
    \end{equation*}
    The real even and odd spherical harmonics corresponding to these
    coefficients are defined as
    \begin{equation}
        Y_{\ell m}^e(\theta,\varphi)
        \;\triangleq\;
        P_\ell^m(\cos\theta)\cos(m\varphi),
        \qquad
        Y_{\ell m}^o(\theta,\varphi)
        \;\triangleq\;
        P_\ell^m(\cos\theta)\sin(m\varphi),
    \end{equation}
    
    Substituting this representation into \eqref{eq:dimensionless_ext_solution},
    the exterior gravitational potential takes the form
    \begin{equation}\label{eq:gravity_field_expansion_final}
        U(r,\theta,\varphi)
        =
        \frac{GM}{r}
        \left[
            1
            +
            \sum_{\ell=1}^{\infty}
            \sum_{m=0}^{\ell}
            \left( \frac{R}{r} \right)^{\ell}
            \left(
                C_{\ell m}\,Y_{\ell m}^e(\theta,\varphi)
                +
                S_{\ell m}\,Y_{\ell m}^o(\theta,\varphi)
            \right)
        \right].
    \end{equation}

\subsection{Recursion Methods for Legendre Polynomials}

The direct computation of the associated Legendre polynomials $P_\ell^m(u)$ becomes computationally expensive for large degree $\ell$ and order $m$, since they are defined through repeated differentiation of lower-order terms as in \cref{eq:polar_solution}. To obtain a numerically efficient and stable evaluation, we employ a set of recursion relations derived in \cite{cunningham_computation_1970}. Let $u \in [-1,1]$ denote the argument of the associated Legendre polynomials (e.g., $u=\sin\theta$ or $u=\cos\theta$, depending on parity). Then, the recursion starts from the base value
\begin{equation}
    P_0^0(u) = 1 .
\end{equation}

For all terms such that $m \ge 1$ and $\ell = m$ (i.e., the diagonal elements), the associated Legendre polynomials satisfy the recursion
\begin{equation}
    P_m^m(u)
    =
    (2m-1)\sqrt{1-u^2}\, P_{m-1}^{m-1}(u).
    \label{eq:legendre_diag_recursion}
\end{equation}

Once the diagonal elements are known, the first off-diagonal terms
($\ell = m+1$) follow from
\begin{equation}
    P_{m+1}^m(u)
    =
    (2m+1)\,u\, P_m^m(u).
    \label{eq:legendre_first_off_diag}
\end{equation}

For degrees $\ell > m+1$, the associated Legendre polynomials satisfy the
three-term recursion
\begin{equation}
    P_\ell^m(u)
    =
    \frac{1}{\ell - m}
    \left[
        (2\ell-1)\,u\, P_{\ell-1}^m(u)
        -
        (\ell+m-1)\, P_{\ell-2}^m(u)
    \right].
    \label{eq:legendre_general_recursion}
\end{equation}

Together, \cref{eq:legendre_diag_recursion,eq:legendre_first_off_diag,eq:legendre_general_recursion} allow the systematic computation of all $P_\ell^m(u)$ with $0 \le m \le \ell \le \ell_{\max}$ using only algebraic operations. These recursions are known to be numerically stable, since either the degree $\ell$ or the order $m$ is held fixed at each step.

\subsection{Computing the Harmonic Coefficients}

The gravitational potential expressed in \cref{eq:gravity_field_expansion_final}, while capable of modeling higher-order gravitational effects, relies on prior knowledge of the harmonic coefficients $C_{\ell m}$ and $S_{\ell m}$ of the central body. These coefficients are body-specific, as they depend on the physical properties of the object, including its geometry and internal mass distribution. Several methods exist for determining the harmonic coefficients, depending on the type and availability of observational data. For a brief review of Earth's geopotential measurement methods, see \cite{sanso_combination_1989}. The most common approach consists of analyzing the motion of spacecraft orbiting the main body and characterizing the perturbations induced by deviations from a purely monopole gravitational field. In this framework, the higher-order terms of \cref{eq:gravity_field_expansion_final} are used as a regression model to fit the observed orbital perturbations and estimate the corresponding harmonic coefficients.

A major drawback of this method is that it requires precise tracking of controlled orbital maneuvers around the central body, which may be impractical or infeasible for objects that are difficult to access, such as small asteroids. In such cases, an alternative approach relies on visual and spectral observations to reconstruct the surface geometry and infer material properties of the body. These data are then used to build a polyhedral shape model, which provides a high-fidelity, closed-form approximation of the gravitational potential. Synthetic spacecraft trajectories are subsequently generated within this gravitational field, and a regression procedure analogous to the one described above is employed to recover the harmonic coefficients \cite{werner_spherical_1997,penarroya_orbit_2019}.

As an alternative to data-driven fitting approaches, one may introduce simplifying assumptions on the geometry and density distribution of the main body in order to derive closed-form expressions for the harmonic coefficients. This is precisely what ellipsoid models do. Assume that the central body has a uniform density distribution and an ellipsoidal geometry with semi-major axes satisfying $0 < c \leq b \leq a$. The harmonic coefficients for this type of geometry exhibits a highly sparse structure. In particular, all sine coefficients vanish, $S_{\ell m} = 0$, and all cosine coefficients with either odd degree $\ell$ or odd order $m$ are identically zero, $C_{\ell m} = 0$. As a result, the only nonzero coefficients are of the form $C_{2\ell,\,2m}$ where $\ell,m \in \mathbb{Z}_{\ge 0}$. Closed-form expressions for these coefficients are presented in \cite{balmino_gravitational_1994}. Let $R_0 \triangleq \left( \frac{3}{a + b + c} \right)^{\frac{1}{2}}$, then the first five nonzero harmonic coefficients are
\begin{align*}
    C_{2,0} &= \frac{1}{5R_0^2} \left( c^2 - \frac{a^2 + b^2}{2} \right),\\
    C_{2,2} &= \frac{1}{20R_0^2} \left( a^2 - b^2 \right), \\
    C_{4,0} &= \frac{15}{7} \left( C_{2,0}^2 + 2 C_{2,2}^2 \right), \\
    C_{4,2} &= \frac{5}{7} C_{2,0} C_{2,2}, \\
    C_{4,4} &= \frac{5}{28} C_{2,2}^2.
\end{align*}

\subsection{Spherical Harmonics in Cartesian Coordinates}

This section introduces a reparametrization of \cref{eq:gravity_field_expansion_final} in cartesian coordinates, including recurrence relations necessary for the Legendre polynomials in cartesian coordinates and an expression for the gradient of the potential (i.e. the gravitational acceleration) based on partial accelerations. This section follows \cite[Chapter 3.2]{montenbruck_satellite_2000} and the main results presented here were derived in \cite{cunningham_computation_1970}. Note that in the formulation of \cite{cunningham_computation_1970,montenbruck_satellite_2000}, the associated Legendre polynomials $P_{\ell}^m(\cos \theta)$ are written as functions of $\sin\phi$. These conventions are equivalent, since $\cos\theta = \sin\phi$. Thus, the use of $P_{\ell}^m(\sin\phi)$ does not alter the angular dependence of the spherical harmonics, but reflects a coordinate choice adapted to Cartesian representations.

In order to facilitate efficient evaluation of the gravitational potential and its derivatives in Cartesian coordinates, consider the solid spherical harmonics $V_{\ell m}$ and $W_{\ell m}$, which combine radial and angular dependencies into homogeneous functions of degree $-(\ell+1)$. They are defined as
\begin{align}
    V_{\ell m}(r,\theta,\varphi)
    &\triangleq
    \left( \frac{R}{r} \right)^{\ell+1}
    P_\ell^m(\cos\theta)\cos(m\varphi),\\
    W_{\ell m}(r,\theta,\varphi)
    &\triangleq
    \left( \frac{R}{r} \right)^{\ell+1}
    P_\ell^m(\cos\theta)\sin(m\varphi),
\end{align}
Note that $V_{0,0} = \frac{R}{r}$ and $W_{0,0} = 0$. The solid spherical harmonics satisfy the following recurrence relations. For $m>0$ and $\ell=m$,
\begin{align}
    V_{m,m}(p)
    &=
    (2m-1)
    \left(
        \frac{xR}{r^2} V_{m-1,m-1}
        -
        \frac{yR}{r^2} W_{m-1,m-1}
    \right), \label{eq:V_mm}
    \\
    W_{m,m}(p)
    &=
    (2m-1)
    \left(
        \frac{xR}{r^2} W_{m-1,m-1}
        +
        \frac{yR}{r^2} V_{m-1,m-1}
    \right). \label{eq:W_mm}
\end{align}
For $\ell \ge m+1$, the general recursion is given by
\begin{align}
    V_{\ell m}(p)
    &=
    \frac{2\ell-1}{\ell-m}
    \frac{zR}{r^2}
    V_{\ell-1,m}
    -
    \frac{\ell+m-1}{\ell-m}
    \left( \frac{R}{r} \right)^2
    V_{\ell-2,m}, \label{eq:V_lm}
    \\
    W_{\ell m}(p)
    &=
    \frac{2\ell-1}{\ell-m}
    \frac{zR}{r^2}
    W_{\ell-1,m}
    -
    \frac{\ell+m-1}{\ell-m}
    \left( \frac{R}{r} \right)^2
    W_{\ell-2,m}. \label{eq:W_lm}
\end{align}

These solid spherical harmonics recursion use the Cartesian parametrization of the position. Which allows us to formulate the gravitational potential \cref{eq:gravity_field_expansion_final} in Cartesian coordinates as
\begin{equation}
    U(p)
    =
    \frac{GM}{R}
    \sum_{\ell=0}^{\infty}
    \sum_{m=0}^{\ell}
    \Big(
        C_{\ell m}\, V_{\ell m}(p)
        +
        S_{\ell m}\, W_{\ell m}(p)
    \Big),
    \label{eq:cartesian_gravity_potential}
\end{equation}

Finally, the gravitational acceleration can be directly calculated from \cref{eq:conservative_potential}, which can be expressed as a superposition of partial accelerations associated with each coefficient,
\begin{equation}
    F_x = \sum_{\ell=0}^{\infty}\sum_{m=0}^{\ell} F_{x_{\ell m}},
    \qquad
    F_y = \sum_{\ell=0}^{\infty}\sum_{m=0}^{\ell} F_{y_{\ell m}},
    \qquad
    F_z = \sum_{\ell=0}^{\infty}\sum_{m=0}^{\ell} F_{z_{\ell m}}.
    \label{eq:acceleration_decomposition}
\end{equation}

For the zonal coefficients ($m=0$), the partial accelerations reduce to
\begin{align}
    F_{x_{\ell 0}}
    &=
    -\,\frac{GM}{R^2}\, C_{\ell 0}\, V_{\ell+1,1},
    \\
    F_{y_{\ell 0}}
    &=
    -\,\frac{GM}{R^2}\, C_{\ell 0}\, W_{\ell+1,1},
    \\
    F_{z_{\ell 0}}
    &=
    \frac{GM}{R^2}\,(\ell+1)
    \left(
        - C_{\ell 0}\, V_{\ell+1,0}
    \right).
\end{align}

For Tesseral and sectorial terms $m>0$, the partial accelerations are given by
\begin{align}
    F_{x_{\ell m}}
    &=
    \frac{GM}{R^2}\,\frac{1}{2}
    \Bigg[
        \big(
            -C_{\ell m} V_{\ell+1,m+1}
            -S_{\ell m} W_{\ell+1,m+1}
        \big)
        \nonumber\\
        &\hspace{2.5cm}
        +\frac{(\ell-m+2)!}{(\ell-m)!}
        \big(
            +C_{\ell m} V_{\ell+1,m-1}
            +S_{\ell m} W_{\ell+1,m-1}
        \big)
    \Bigg],
    \label{eq:xddot_tesseral}
    \\
    F_{y_{\ell m}}
    &=
    \frac{GM}{R^2}\,\frac{1}{2}
    \Bigg[
        \big(
            -C_{\ell m} W_{\ell+1,m+1}
            +S_{\ell m} V_{\ell+1,m+1}
        \big)
        \nonumber\\
        &\hspace{2.5cm}
        +\frac{(\ell-m+2)!}{(\ell-m)!}
        \big(
            -C_{\ell m} W_{\ell+1,m-1}
            +S_{\ell m} V_{\ell+1,m-1}
        \big)
    \Bigg],
    \label{eq:yddot_tesseral}
    \\
    F_{z_{\ell m}}
    &=
    \frac{GM}{R^2}\,(\ell-m+1)
    \big(
        -C_{\ell m} V_{\ell+1,m}
        -S_{\ell m} W_{\ell+1,m}
    \big).
    \label{eq:zddot_tesseral}
\end{align}

The following algorithm summarizes the steps to compute the gravitational potential and the gravitational acceleration exerted by the main body whose harmonic coefficients $\{C_{\ell m}, S_{\ell m}\}$ are known to a degree $\ell_{\max}$ for a given position in cartesian coordinates.

\begin{algorithm}[H]
\DontPrintSemicolon
\caption{High-Order Gravitational Potential and Acceleration in Cartesian Coordinates}
\label{alg:cartesian_gravity_2e}

\KwIn{Position $(x,y,z)$, reference radius $R$, harmonic coefficients $\{C_{\ell m}, S_{\ell m}\}$ for $0 \le m \le \ell \le \ell_{\max}$}
\KwOut{Gravitational potential $U$ and acceleration $F(r) = (F_x,F_y,F_z)$}

Compute $r \gets \sqrt{x^2 + y^2 + z^2}$\;

Initialize $U \gets 0$, $F_x \gets 0$, $F_y \gets 0$, $F_z \gets 0$\;

Initialization of solid spherical harmonics: $V_{0,0} \gets R / r$, $W_{0,0} \gets 0$\;

\For{$m \gets 1$ \KwTo $\ell_{\max}+1$}{
    Compute $V_{m,m}$ using \cref{eq:V_mm} and $W_{m,m}$ using \cref{eq:W_mm}.\;
}

\For{$m \gets 0$ \KwTo $\ell_{\max}+1$}{
    \For{$\ell \gets m+1$ \KwTo $\ell_{\max}+1$}{
        Compute $V_{\ell m}$ using \cref{eq:V_lm} and $W_{\ell m}$ using \cref{eq:W_lm}.\;
    }
}

Compute gravitational potential $U$ using \cref{eq:cartesian_gravity_potential}.

\tcp{\textbf{Compute gravitational acceleration:}}
\For{$\ell \gets 0$ \KwTo $\ell_{\max}$}{
    \tcp{Zonal terms $m=0$}
    Update $(F_x,F_y,F_z)$ using Eqs.~\cref{eq:xddot_tesseral}--\cref{eq:zddot_tesseral} with $m=0$\;

    \For{$m \gets 1$ \KwTo $\ell$}{
        \tcp{Tesseral and sectorial terms $m>0$}
        Update $(F_x,F_y,F_z)$ using Eqs.~\cref{eq:xddot_tesseral}--\cref{eq:zddot_tesseral}\;
    }
}

\Return $U$, $(F_x,F_y,F_z)$\;

\end{algorithm}

\newpage
\section{Studying Orbits of Highly Perturbed Systems}

In this final section, we present two examples where the study of higher-order gravitational effects is particularly significant.

To facilitate these studies, we developed a Python library \textbf{ARC-Grav}\footnote{Source code available at \url{https://github.com/FelipeArenasUribe/ARC-Grav}.}, which enables flexible and efficient modeling of a body's gravitational field. The library is structured around a concept of \textit{main bodies} composed of one or more \textit{sub-bodies}, which can be spherical, ellipsoidal, or uniform. For spherical and ellipsoidal sub-bodies, the library requires only the physical parameters of the body, from which the normalized harmonic coefficients are automatically computed. For uniform sub-bodies, the normalized harmonic coefficients must be provided explicitly. This modular approach allows users to easily create environments with main bodies composed of multiple sub-bodies, making it possible to model real and complex gravitational systems with minimal setup.

\subsection{Low Earth Orbit}

Low Earth Orbit (LEO) provides a particularly sensitive environment for studying higher--order gravitational effects of the Earth. The majority of artificial satellites currently operate in orbits ranging from LEO to geostationary Earth orbit (GEO), at altitudes between $200-500$ kilometers. These satellites are influenced not only by the dominant central gravitational attraction, but also by deviations from spherical symmetry arising from the Earth’s oblateness and heterogeneous mass distribution. These perturbations accumulate over time and lead to deviations in orbital motion. As a result, satellites in these regimes are required to perform periodic station--keeping maneuvers in order to maintain their prescribed orbital parameters. Consequently, understanding how satellite orbits are perturbed by high--order gravitational effects is of significant practical and scientific interest.

In this example, we simulate the motion of a satellite in Low Earth Orbit (LEO) using two gravitational models: a Newtonian model of the Earth and a higher--order model based on spherical harmonics from the Joint Gravity Model JGM-3 \cite{tapley_joint_1996}. The Earth is modeled with mass $M_E = 5.972\times10^{24}\,\mathrm{kg}$ and mean radius $R_E = 6378.1363\,\mathrm{km}$. Three circular LEO trajectories at an altitude of $500\,\mathrm{km}$ are considered: an equatorial orbit ($i=0^\circ$), a polar orbit ($i=90^\circ$), and an inclined orbit ($i=45^\circ$). Each orbit is propagated for ten orbital periods using a fixed time step of $\Delta t = 10\,\mathrm{s}$.

\begin{figure}[ht]
    \centering
    \includegraphics[width=0.5\linewidth]{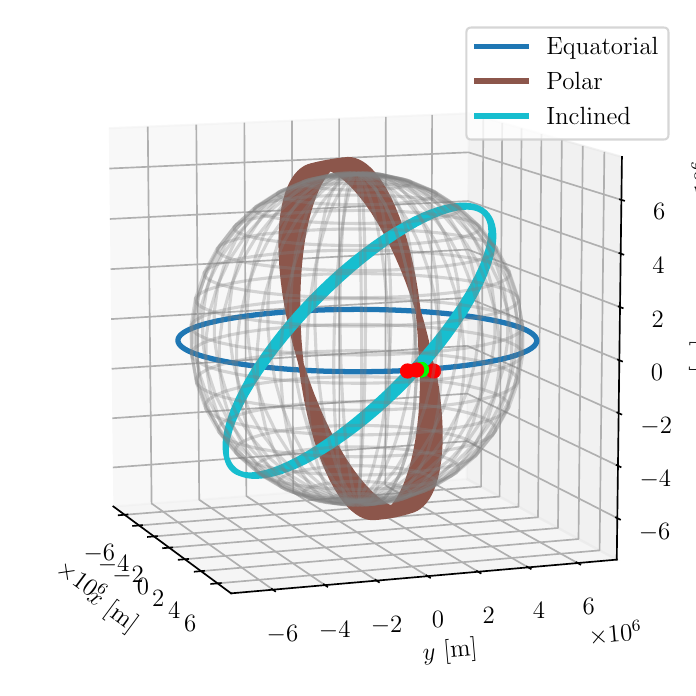}
    \caption{LEO trajectories with JGM-3 Higher-Order model.}
    \label{fig:LEO_trajectories}
\end{figure}

As shown in \cref{fig:Equatorial_LEO}, the equatorial orbit follows essentially the same trajectory in the $xy$--plane when modeled using both the Newtonian and higher--order gravitational formulations. However, the inclusion of higher--order effects introduces a bounded, oscillatory perturbation along the $z$--axis. This behavior becomes more pronounced for the polar orbit shown in \cref{fig:Polar_LEO}, where the perturbation along the $y$--axis is not only oscillatory but also unbounded. As a consequence, \cref{fig:LEO_trajectories} shows that the polar orbit experiences a gradual reduction in its inclination angle, decreasing to approximately $83^\circ$ after ten orbital periods. Finally, the inclined orbit presented in \cref{fig:Inclined_LEO} exhibits comparatively milder perturbations: the trajectory along each axis shows deviations of less than $1000~\mathrm{m}$, which nonetheless accumulate over time, leading to changes in the orbital parameters, including a reduction in the orbital radius and a slight perturbation of the orbital inclination.

\begin{figure}[ht]
    \centering

    \begin{subfigure}[b]{0.30\linewidth}
        \centering
        \includegraphics[width=\linewidth]{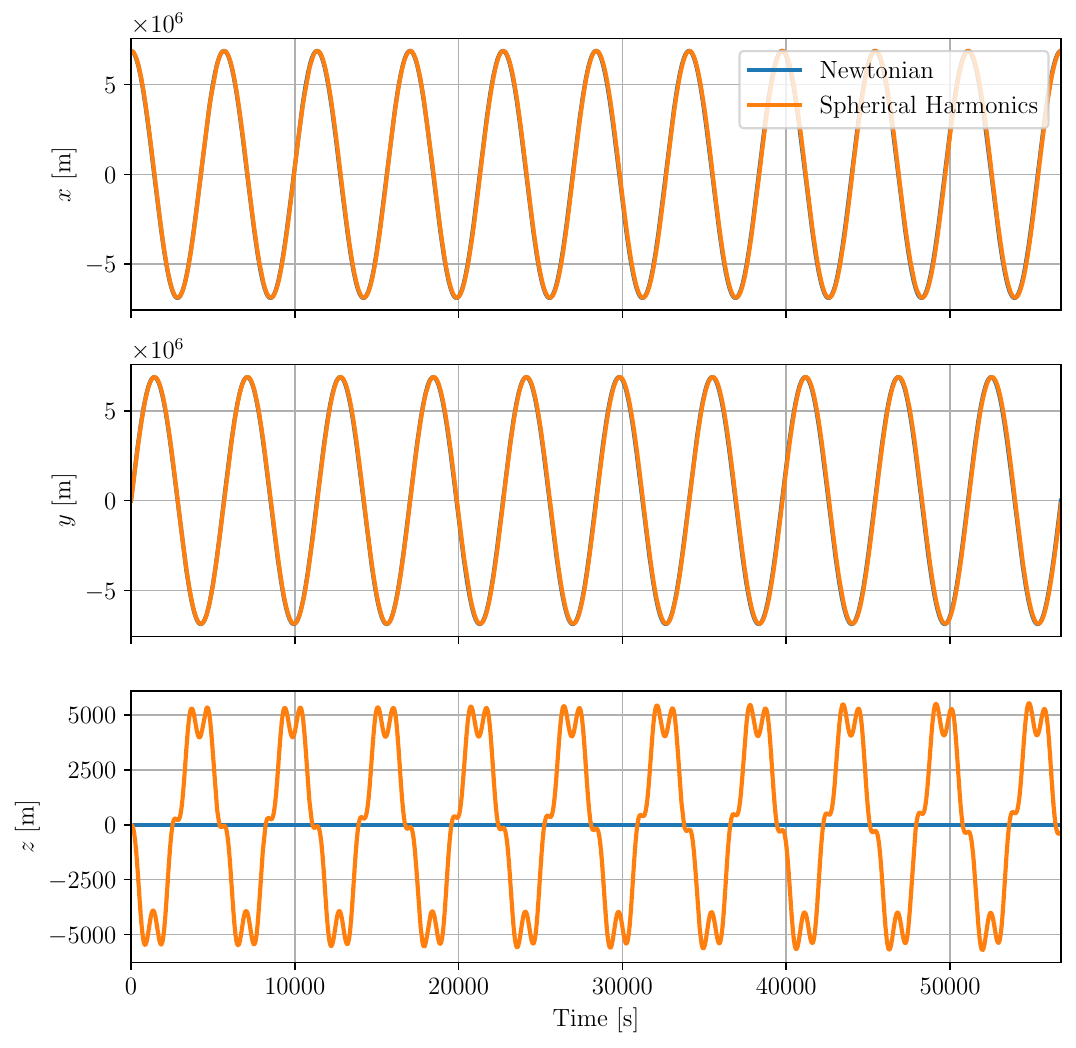}
        \caption{Equatorial LEO}
        \label{fig:Equatorial_LEO}
    \end{subfigure}
    \hfill
    \begin{subfigure}[b]{0.30\linewidth}
        \centering
        \includegraphics[width=\linewidth]{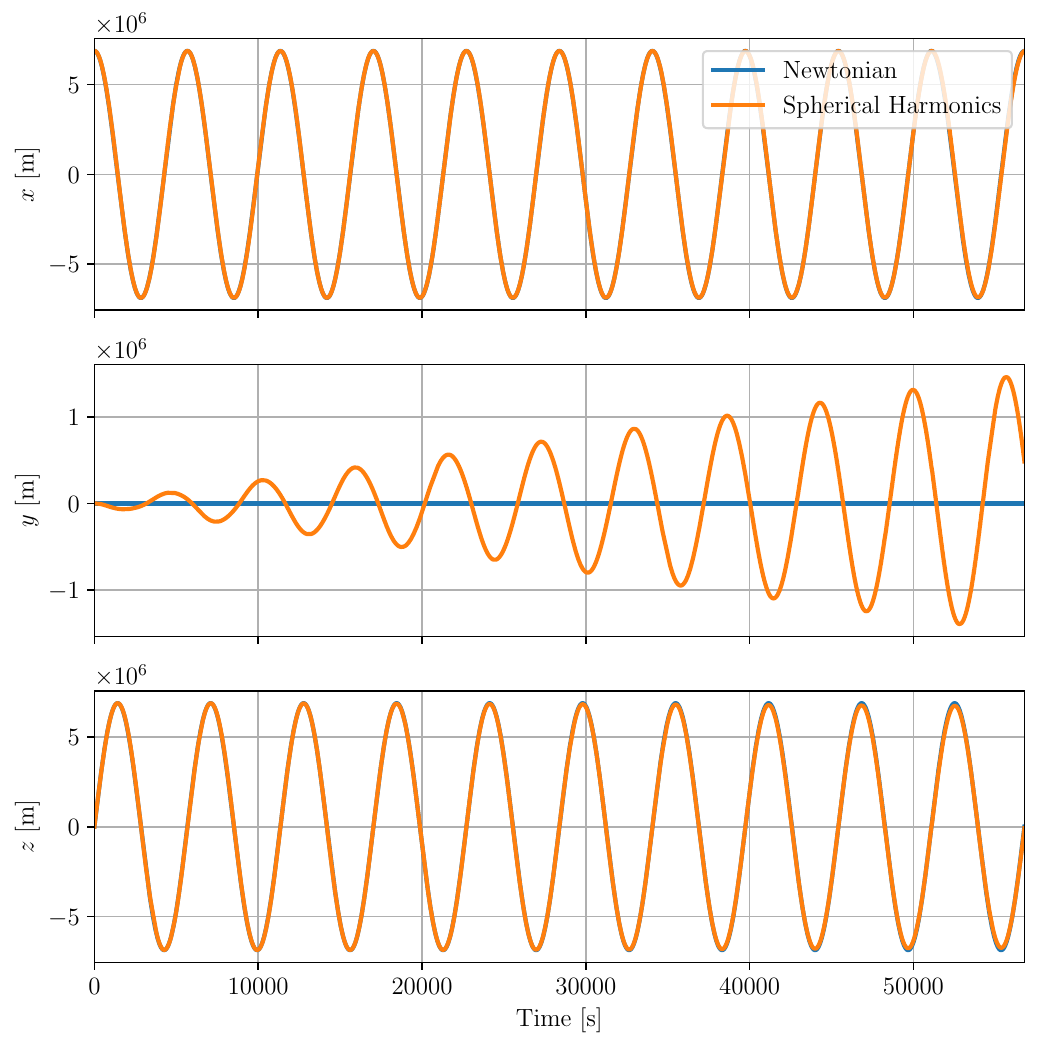}
        \caption{Polar LEO}
        \label{fig:Polar_LEO}
    \end{subfigure}
    \hfill
    \begin{subfigure}[b]{0.30\linewidth}
        \centering
        \includegraphics[width=\linewidth]{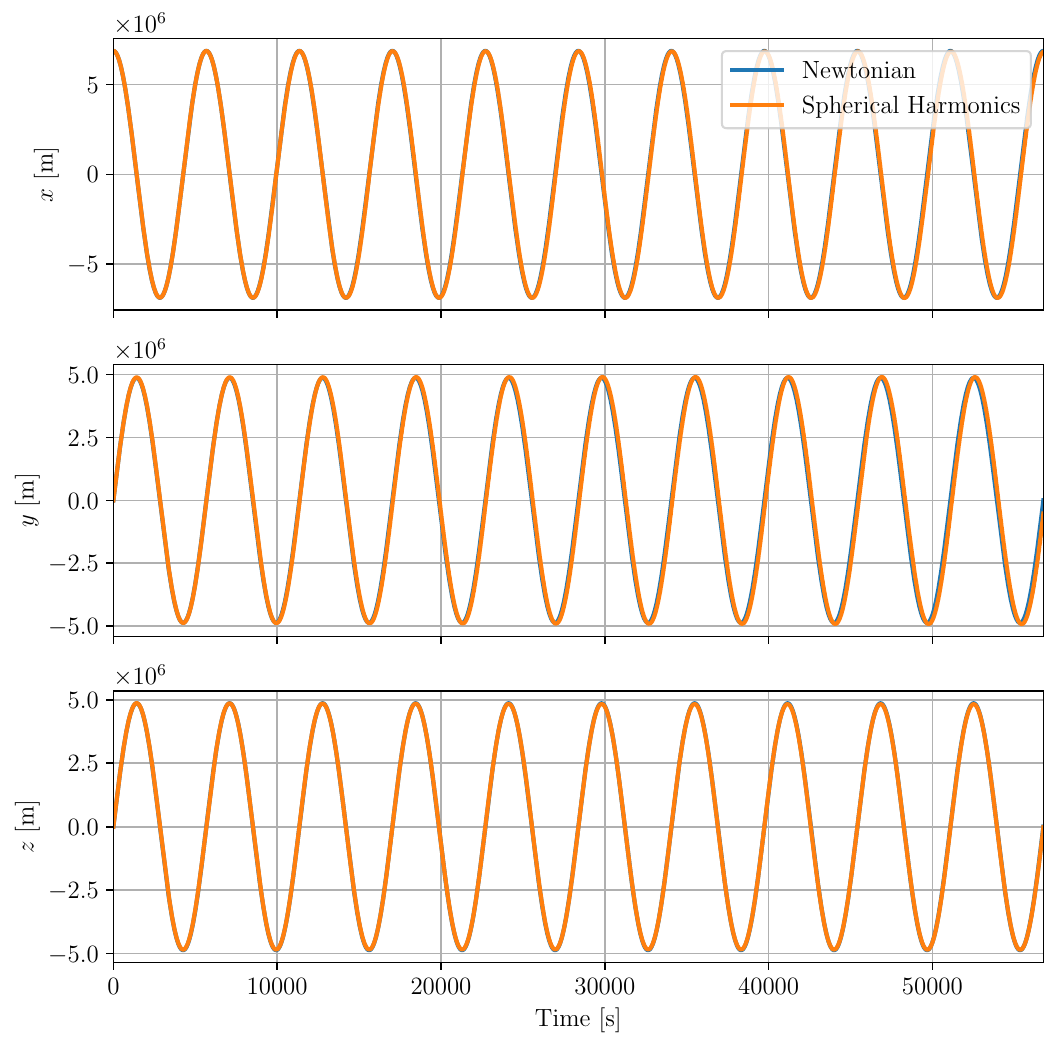}
        \caption{Inclined LEO}
        \label{fig:Inclined_LEO}
    \end{subfigure}

    \caption{Comparison of equatorial, polar, and inclined LEO trajectories under Newtonian and Spherical Harmonic gravitational models.}
    \label{fig:LEO_Comparison}
\end{figure}

\subsection{Irregular Asteroid Modeling}

Small celestial bodies such as asteroids present both scientific and engineering challenges due to their irregular shapes and often non-uniform density distributions. These characteristics result in highly complex gravitational fields, which complicate orbit design, proximity operations, and landing maneuvers. In this example, we develop an analytical model of an irregular asteroid with non-uniform density by representing it as a composition of multiple ellipsoidal sub-bodies. Each ellipsoid is defined by its semi-major axes, uniform density, position, and orientation with respect to the inertial frame. The gravitational field of the asteroid is computed as the superposition of the gravitational contributions of all sub-bodies. This approach enables the construction of flexible and computationally efficient gravitational models for irregular bodies. Although not intended to reproduce a specific asteroid, the method is well suited for generating realistic test cases for orbit design and for training and evaluating control strategies for operations around irregular celestial bodies.

For this example, we define three ellipsoidal components with different sizes and orientations. The first ellipsoid, located at the origin, has dimensions $800~\mathrm{m} \times 300~\mathrm{m} \times 300~\mathrm{m}$ and a density of $1500~\mathrm{kg/m^3}$. The second ellipsoid is offset to $[300, -200, 0]~\mathrm{m}$ and rotated by $45^\circ$ around the $z$-axis, with dimensions $600~\mathrm{m} \times 400~\mathrm{m} \times 200~\mathrm{m}$ and the same density. The third ellipsoid is positioned at $[-250, 150, 100]~\mathrm{m}$, with a rotation of $30^\circ$ about the $y$-axis, and dimensions $500~\mathrm{m} \times 250~\mathrm{m} \times 350~\mathrm{m}$.

\begin{figure}[htbp]
    \centering

    \begin{subfigure}[b]{0.45\linewidth}
        \centering
        \includegraphics[width=\linewidth]{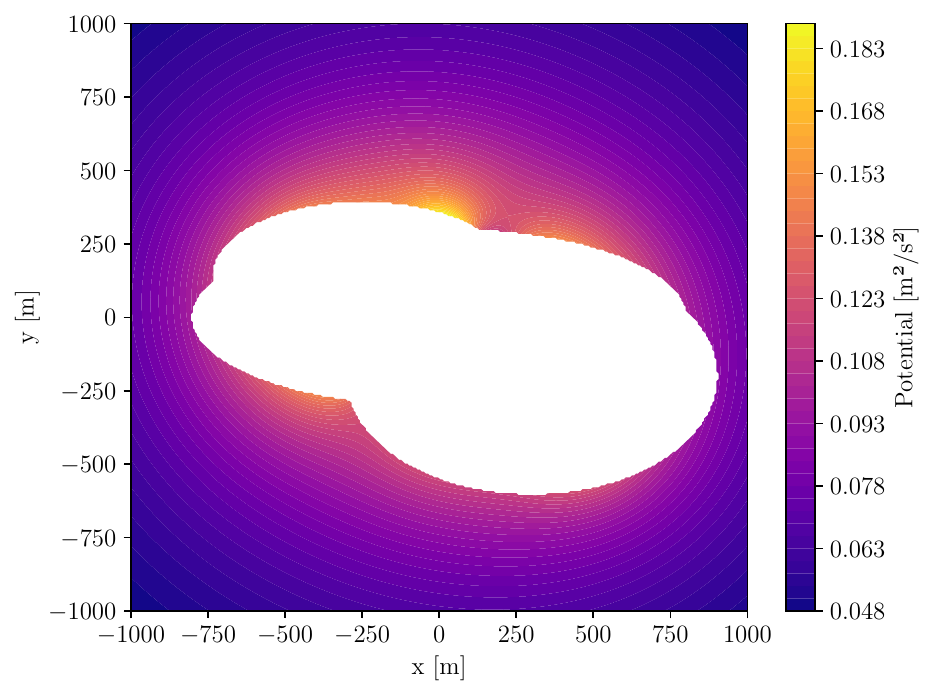}
        \caption{Gravitational Potential of the asteroid.}
        \label{fig:Potential}
    \end{subfigure}
    \hfill
    \begin{subfigure}[b]{0.4\linewidth}
        \centering
        \includegraphics[width=\linewidth]{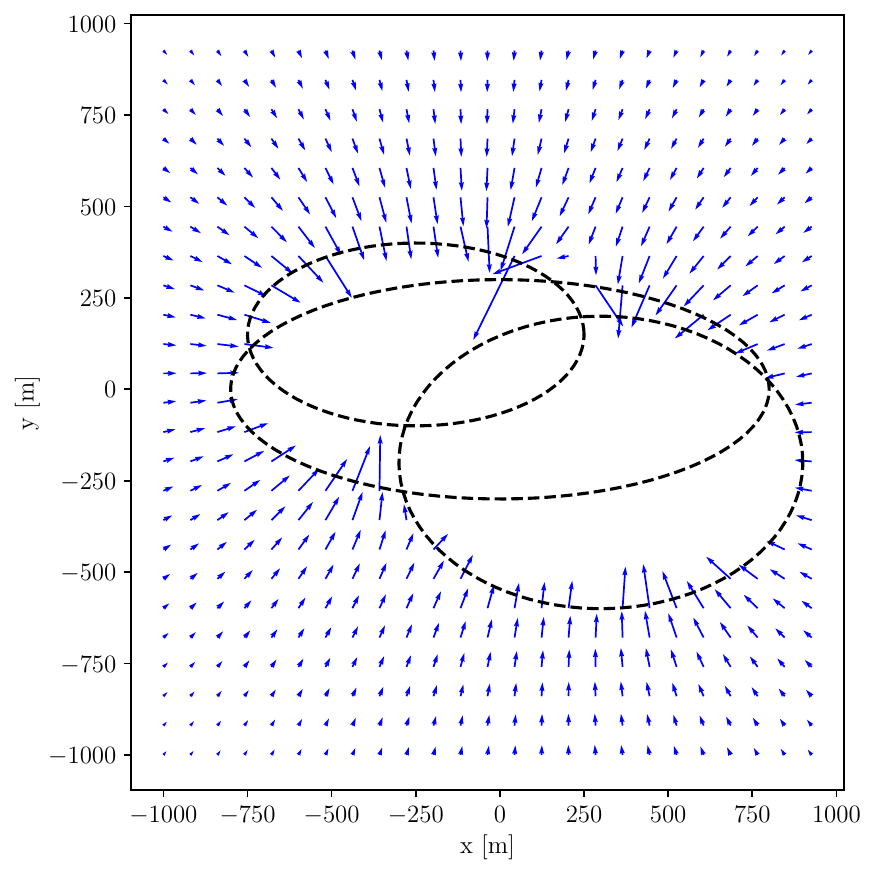}
        \caption{Gravitational field vector.}
        \label{fig:Vector_field}
    \end{subfigure}

    \caption{Properties of the gravitational field of the custom irregularly shaped asteroid with non-uniform density distribution.}
    \label{fig:Asteroid_Grav_field}
\end{figure}

\Cref{fig:Potential} shows the gravitational potential in the $XY$ plane. As expected for an irregularly shaped body composed of multiple ellipsoidal sub-bodies, the potential is highly non-uniform. These variations highlight how the superposition of the sub-bodies creates complex gravitational features that deviate significantly from a simple point-mass or homogeneous ellipsoid approximation. The gravitational acceleration vector field \cref{fig:Vector_field}, provides further insight into the directional nature of the gravitational forces acting on nearby objects. The field vectors exhibit significant asymmetry due to the non-uniform mass distribution and irregular geometry of the asteroid. This non-uniformity can significantly affect landing maneuvers, where errors in the gravitational models diverge from the real gravitational effects.

\cref{fig:asteroid_trajectory} shows the trajectory of a test body initialized at $r_0 = [1200, 0, 0]~\mathrm{m}$ with a low initial velocity strictly in the y-axis. The orbit demonstrates stable behavior around the asteroid despite the irregular gravitational field, indicating that the combined gravitational effect of the ellipsoids produces regions where quasi-stable orbits are feasible. Small deviations from circularity and perturbations in the trajectory reflect the influence of the asteroid's irregular mass distribution. 

\begin{figure}[ht]
    \centering
    \includegraphics[width=0.5\linewidth]{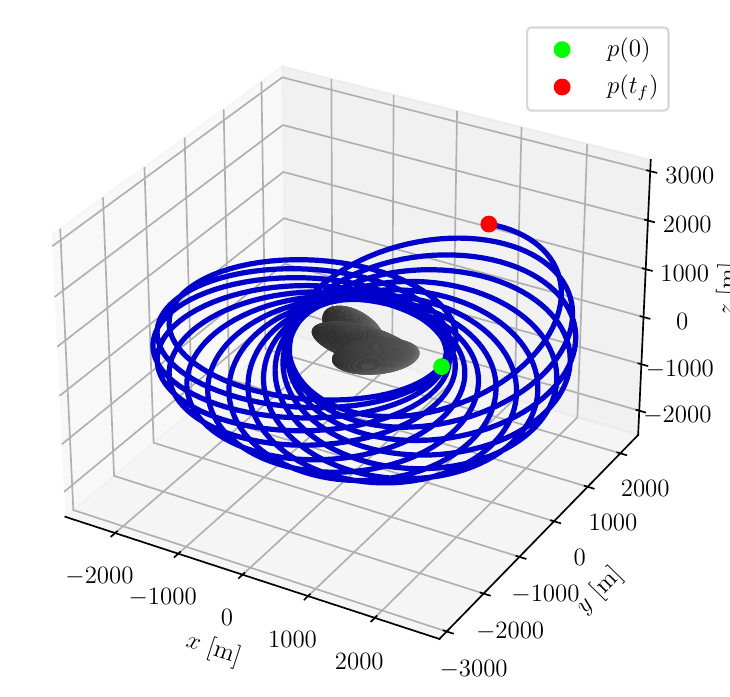}
    \caption{Stable trajectory about the custom asteroid.}
    \label{fig:asteroid_trajectory}
\end{figure}

\newpage
\bibliographystyle{elsarticle-num} 
\bibliography{references}

\end{document}